\let\oldvec\vec

\documentclass[ruled,vlined,linesnumbered,commentsnumbered]{llncs}

\let\vec\oldvec

\usepackage[ruled,vlined,linesnumbered,commentsnumbered]{algorithm2e}
\usepackage{amsmath,amssymb,amsfonts}

\begin{document}

\newtheorem{fact}[lemma]{Fact}

\title{Bounded Degree Group Steiner Tree Problems}

\author{Guy Kortsarz\inst{1} \and Zeev Nutov\inst{2}} 

\institute{Rutger university, Camden. \email{guyk@rutgers.edu} \and The Open University of Israel. \email{nutov@openu.ac.il}}

\maketitle

\newcommand {\ignore} [1] {}

\def\gst   {\sc Group Steiner Tree}
\def\bd   {\sc Bounded Degree}
\def\md   {\sc Min-Degree}
\def\skt   {\sc Steiner $k$-Tree}


\def\TT     {{\cal T}}
\def\XX     {{\cal X}}
\def\SS     {{\cal S}}
\def\FF     {{\cal F}}
\def\AA    {{\cal A}}
\def\LL    {{\cal L}}
\def\CC    {{\cal C}}

\def\si     {\sigma}
\def\Ga   {\Gamma}
\def\de     {\delta}
\def\al     {\alpha}

\def\empt {\emptyset}
\def\sem  {\setminus}
\def\subs  {\subseteq}

\def\t   {\tilde}

\def\f   {\frac}
\def\tw {\mathsf{tw}}

\begin{abstract}
We study two problems that seek a subtree $T$ of a graph $G=(V,E)$ 
such that $T$ satisfies a certain property and has minimal maximum degree.
\begin{itemize}
\item
In the {\md} {\gst} problem we are given a collection $\SS$ of groups (subsets of $V$) and $T$ should contain a node from every group.
\item
In the {\md} {\skt} problem we are given a set $R$ of terminals and an integer $k$,
and $T$ should contain at least $k$ terminals.
\end{itemize}
We show that if the former problem admits approximation ratio $\rho$ then the later problem admits approximation ratio $\rho \cdot O(\log k)$.
For bounded treewidth graphs, we obtain approximation ratio $O(\log^3 n)$ for {\md} {\gst}.

In the more general {\bd} {\gst} problem we are also given edge costs and degree bounds $\{b(v):v \in V\}$,
and $T$ should obey the degree constraints $\deg_T(v) \leq b(v)$ for all $v \in V$.
We give a bicriteria $(O(\log N \log |\SS|),O(\log^2 n))$-approximation algorithm for this problem on tree inputs,
where  $N$ is the size of the largest group, 
generalizing the approximation of Garg, Konjevod, and Ravi \cite{GKR} for the case without degree bounds.
\end{abstract}

\section{Introduction}

All graphs in the paper are assumed to be undirected, unless stated otherwise.
We study two problems that seek a subtree $T$ of a given graph $G=(V,E)$ 
such that $T$ satisfies a certain property and has minimal maximum degree.

\begin{center} \fbox{\begin{minipage}{0.97\textwidth} \noindent
\underline{{\md} {\gst}} \\
Input: \hspace*{0.2cm} A graph $G=(V,E)$ and a collection ${\cal S}$ of groups (subsets of $V$). \\
Output: A subtree $T$ of $G$ that contains a node from every group and has minimal maximum degree. 
\end{minipage}}\end{center} 

\begin{center} \fbox{\begin{minipage}{0.97\textwidth} \noindent
\underline{{\md} {\skt}} \\
Input: \hspace*{0.2cm} A graph $G=(V,E)$, a set $R \subs V$ of terminals, and an integer $k \leq |R|$. \\
Output: A subtree $T$ of $G$ that contains at least $k$ terminals and has minimal maximum degree. 
\end{minipage}}\end{center} 
 
The best ratio known for {\sc Min-Cost} {\gst} on tree inputs is $O(\log N \log |\SS|)=O(\log^2 n)$  
where $N$ is the size of the largest group \cite{GKR};
there is also combinatorial algorithm with ratio $O(\log^{2+\epsilon} n)$ \cite{CEK}. 
In the case of general graph inputs, embedding the 
input graph into a tree distribution with stretch $O(\log n)$ \cite{FRT} gives 
ratio by a factor $O(\log n)$ larger.
The above ratio for tree inputs is essentially tight due to the 
approximation threshold $\Omega(\log^{2-\epsilon}n)$ of \cite{HK}. 
It is easy to see that {\md} {\gst} is {\sc Hitting Set}/{\sc Set Cover} hard even on stars. 
Given a {\sc Hitting Set} instance, add a root $r$ connected to every element, and define the groups to be the sets.
Then any substar corresponds to a hitting set of size equal to the degree of $r$ in the substar. 

{\sc Min-Cost} {\skt} generalizes the $k$-{\sc MST} problem (the case $R=V$) that admits ratio $2$ \cite{G};
this implies ratio $4$ for {\sc Min-Cost} {\skt}.
The directed version of {\md} {\skt} admits ratio $O(\sqrt{n/d^*})$ \cite{KKN} 
where $d^*$ is the minimum possible degree.
No better approximation is known for undirected graphs.

Our results for these two problems are summarized in the following two theorems.

\begin{theorem} \label{t:1}
If {\md} {\gst} admits approximation ratio $\rho$, then {\md} {\skt} admits ratio $\rho\cdot O( \log k)$.
\end{theorem}

\begin{theorem} \label{t:2}
{\md} {\gst} on bounded treewidth input graphs admits approximation ratio $O(\log^3 n)$.
\end{theorem}

On bounded tree width graphs {\md} {\skt} admits a polynomial time algorithm by dynamic programming,
so in this particular case a direct combination of Theorems \ref{t:1} and \ref{t:2} does not give a useful result.
However, Theorem \ref{t:1} is of interest for general and other types of graphs.
We became aware that recently {\md} {\gst} was shown to admit a non-trivial ratio on general graphs \cite{BP},
and combined with Theorem~\ref{t:1} this gives the first non trivial ratio for {\bd} {\skt} on general graphs.

We also consider the degree bounded version of the {\sc Min-Cost} {\gst} problem.

\begin{center} \fbox{\begin{minipage}{0.97\textwidth} \noindent
\underline{{\bd} {\gst}} \\
Input: \ A graph $G=(V,E)$ with edge costs $\{c(e):e \in E\}$, a collection ${\cal S}$ of groups,
and degree bounds $\{b(v):v \in V\}$. \\
Output: A min-cost subtree $T$ of $G$ that contains a node from every group and
obeys the degree bounds $\deg_T(v) \leq b(v)$ for all $v \in V$.
\end{minipage}}\end{center} 

\begin{theorem} \label{t:3}
{\bd} {\gst} on tree inputs admits a bicriteria randomized $(O(\log N \log |\SS|),O(\log^2 n))$-approximation algorithm,
where $N$ is the size of the largest group. 
Namely, the algorithm computes a tree $T$ that contains at least one node from every group,
has expected cost $O(\log N \log |\SS|)$ times the optimal, 
and with probability at least $1-1/n$ we have $\deg_T(v)=O(\log^2 n) \cdot b_v$ for all $v \in V$.
\end{theorem}

This result generalizes the result of Garg, Konjevod, and Ravi \cite{GKR} for the case without degree bounds.
A bicriteria approximation $(O(\log N \log |\SS|),O(\log^3 n))$ for {\bd} {\gst} on tree inputs can be 
deduced from \cite{GKR}, but obtaining the better ratio in Theorem~\ref{t:3} is non-trivial. 
We also note that our result does not extend to general graph using known tree embeddings, 
since these may considerably increases the degrees.

In the rest of this section we give some motivation to these problems and review related work. 
Bounded and minimum degree network design problems have a wide range of applications, 
and they were studied extensively since \cite{J,BKN,LNSM}; see also the book \cite{LRS}. 
For many classic degree bounded problems, good approximation ratios were achieved using the Iterative Rounding method.
However, this method does not seem applicable for the problems we consider, 
e.g., for {\md} {\skt}, as is mentioned in \cite{LRS}.

The main motivation for the {\gst} problem comes from VLSI design.
The goal is to connect a collection $S\subseteq V$ of terminals to a designated root $r$ by a min-cost tree.
Any terminal has a set of multiple ports it can be placed at (ports of two different terminals may intersect).
The set of different ports in which a terminal may be placed at defines a group. 
The different possible location may be due to rotating,  or  mirroring, or  both. 
While low cost is highly desirable, the cost is payed once, and later the VLSI circuit is applied constantly.
Low degree implies that the computations can be done fast.
In \cite{china}, a natural VLSI problem reduces to our problem. 
This paper builds a tree with low degrees in order to bound from above the latency of the VLSI computation.
Low degrees are also important for efficient layout of the VLSI circuit \cite{SK}.
In the {\sc Multicommodity Facility Location under Group Steiner Access} problem \cite{PR}, 
each facility belongs to a group Steiner tree.  Short service times requires that such trees have low degrees.

Our motivation for studying the {\md} {\skt} problem comes from the 
{\sc Telephone $k$-Multicast} problem \cite{SCH}.
In this problem we are given an undirected graph 
and a vertex $r$ and a target $k$ terminals.
We want to send a message known by some root $r$ to at least $k$ terminals.
The communication model is the telephone model. In this model, 
the vertices that know the message can call
at most one neighbor and send it the message.
This means that a round is a matching between vertices who know the 
message to vertices who do not.
Say that $r$ knows the message. The {\sc Telephone $k$-Multicast} problem \cite{SCH} 
is to notify at least $k$ of the terminals in minimum number of rounds.
Note that every broadcasting scheme results in a directed tree 
in which your parent is the one that sent you the message.
The maximum degree in this multicast tree is a lower bound on the optimum,
because at every round we can send the message to at most one child.
Hence we need trees with $k$ terminals and low Maximum degree.

In addition we point out that 
the {\skt} problem is the minimum degree variant of two important and well studied problems:
the $k$-MST and the $k$-{\sc Steiner Tree} problems (see  \cite{G}).
As these problems are important, so are their their minimum degree variants.

\section{Proof of Theorem~\ref{t:1}}

Assume that {\md} {\gst} admits ratio $\rho$.
We will show that then {\md} {\skt} admits ratio $\rho\cdot O(\log k)$.
Fix an optimal solution $T^*$ for the {\md} {\skt} instance with maximum degree $d^*$ and terminal set $R^*$.

\medskip

We first give a simple randomized algorithm with expected ratio $\rho \cdot O(\log^2 k)$.
Given a {\md} {\skt} instance $G,R,k$, create $k/(5\log k)$ bins; the {\md} {\gst} instance groups collection $\SS$ is formed by 
putting uniformly at random, each terminal to a random bin. 

\begin{lemma} \label{l:pr}
With probability at least $1-1/k$ each bin contains a node from $R^*$.
\end{lemma}
\begin{proof}
For each $S \in \SS$, $|S \cap R^*|$
is a binomial variable with probability $5\log k/k$ and $k$ trials. 
Thus the expected size of $|S \cap R^*|$ is $\mu=5\log k$. By the Chernoff bound (c.f. \cite{MR})
$$
\Pr\left[|S \cap R^*| \leq (1-\delta)\mu\right] \leq \exp(-\delta^2\mu/2) \ .
$$
We plug the right $\delta$ so that $(1-\delta)\mu \leq 1$. This gives $\delta$ very close to $1$.
By the Chernoff bound $\Pr[S \cap R^*=\empt] \leq 1/k^2$. By the union bound 
we get that with probability at least $1-1/k$ each bin contains a node from $R^*$.
\qed
\end{proof}

The {\bd} {\gst} solution can choose the tree $T^*$ for the $k$-tree version to span $R^*$. 
Note that we need to cover only $k/(5\log k)$ groups which is {\em not} the {\md} {\gst} problem.
However, here is a trivial reduction to the {\md} {\gst} problem.
Attach a complete binary tree to the root, with $k-k/5\log k$ new leaves
(we may need to trim the tree to get exactly $k-k/(5\log k)$ leaves). 
Every new leaf belongs to all groups.
Thus $k-k/(5\log n)$ groups are covered for free with maximum degree $3$. 
This still requires covering  $k/(5\log k)$ new terminals completing the reduction.
The assumed algorithm will find  a tree containing at least 
$k/5\log k$ terminals, with maximum degree bounded by $\rho \cdot d^*$.
Taking $O(\log^2 k)$ iterations gives expected ratio $O(\log^2 k\cdot \rho)$.

\medskip

We now describe a more complicated determenistic reduction with factor loss $\log k$ in the ratio. 
Let the terminals be $0,1,\ldots,q-1$, $q>k$, and assume that the above order of the terminals is random.
We build $k$ bins to serve as groups using
two point based sampling (see \cite{MR}).
Let $p$ be a prime such that $2k\leq p\leq 4k$.
\begin{enumerate}
\item 
Choose a number $a$, at random, from $1,2,\ldots,p-1$.
\item 
Choose a number $b$, at random, from $0,1,\ldots,p-1$. 
\item 
Terminal $0\leq i\leq q-1$  is assigned bin $((ai+b)\mod p)\mod k.$
\end{enumerate}

Any true terminal $i$ is first matched to a {\em random number} in
$0,1,\ldots,p$. The values that will cause item $i$ to reach bin $j$ 
are $j,j+k,\ldots,j+\alpha\cdot k$ for the maximum 
integer
$\alpha$ such that $\alpha\cdot k\leq p-1$.
In the worst case $j=k-1$. 
Thus the question is how large is $\alpha$ in the inequality
$(k-1)+\alpha\cdot k\leq p-1$. Choosing $\alpha=(p-k)/k$
achieves the desired bound.
Since $\alpha$ is an integer, clearly, 
$p/k-2\leq \alpha<p/k$. Dividing by $p$, implies that 
the probability that terminal $i$ reaches bin $j$ is
at least $1/k-2/p$ and less than $1/k$.

Let $X_{ij}$ be the event that  {\em a true terminal $i$ reaches bin $j$}. 
By the above, $\Pr(X_{ij})\geq 1/k-2/p$.
The events "$i$ arrived to bin $j$" and "$i'$ arrived 
to bin $j$" for $i\neq i'$ are pairwise independent and so 
$\Pr(i\mbox{ and } i' \mbox{ arrive to bin }j)\leq 1/k^2$.
We lower bound the probability that $j$ is full.
using the first two terms of the inclusion exclusion formula
$$
\Pr\left[\bigcup_{i=0}^{k-1} X_{ij}\right]\geq k\cdot \left(\frac{1}{k}-\frac{2}{p}\right)-\frac{{k\choose 2}}{k^2}\geq \frac{1}{2}-\frac{2}{p} \ .
$$
Thus for every bin, the probability that it's full is at least $1/3$ and in expectation number of full bins is at least $k/3$.
Hence there exists a pair $a,b$ in the sample space for which at least $k/3$ bins are full.
Fortunately, our sample space of all $a,b$ pairs has size polynomial in $n$.
Thus we try all $a,b$ pairs with the goal of 
covering at least $k/3$ groups (we have shown above 
that this problem of covering $k/3$ groups can be reduced without penalty
in the degrees to the usual {\md} {\gst} problem). 
For every pair $a,b$, we apply the assumed $\rho$ ratio algorithm.
For at least one of the $a,b$ we get a tree with maximum 
degree at most $\rho \cdot d^*$ and at least $k/3$ 
terminals are covered. Thus outputting the minimum 
maximal degree tree over all $a,b$ choices guaranties 
(with probability $1$) that the maximum degree in the tree is at most $\rho\cdot d^*$, 
and at least $k/3$ groups  are covered.
The penalty is an additional $O(\log n)$ factor (on top of the $\rho$ factor).

\section{Proof of Theorem~\ref{t:2}}

We start by defining the treewidth of a graph.

\begin{definition}  
A {\bf tree decomposition} of a graph $G=(V,E)$ is a tree ${\cal T}$ on a collection 
${\cal X}$ of subsets of $V$, called {\bf bags} such that: 
{\em (i)} for every $uv \in E$ some $X \in {\cal X}$ contains both $u$ and $v$;
{\em (ii)} for every $v \in V$, the bags that contain $v$ induce a subtree of ${\cal T}$.
The {\bf width of a tree decomposition $\XX,\TT$} is $\displaystyle \max_{X \in {\cal X}}|X|-1$. 
The {\bf treewidth of a graph $G$} denoted by ${\tw}(G)$ is the smallest width of a tree decomposition of $G$. 
\end{definition}

Our algorithm does not use the above definition.
Instead, we use the fact that bounded treewidth graphs 
have an $O(1)$ separator. We build a tree of separators 
by computing the separator for the graph, computing the connected 
components resulting from removing the separator, and recursing on
each connected component.
Then we find a way to connect the separators that does not increases
the degrees by much. Hence we can contract each separator into a single node.
The groups of a contracted set is the union of all the groups of the set contracted.
This results in a DFS tree (namely all non tree edges are backward edges) of height $O(\log n)$.
If the optimal solution uses backward edges, we show how to change 
optimum into a new solution that does not contain backward edges,
with additive penalty of $O(\log n\cdot d^*)$ on the degrees.
Let $G'$ be the graph resulting after the contraction of the connected separators, and 
let $T'$ be the tree resulting by removing all the backward edges of $G'$.
Note that the modified solution is a subtree of $T'$. This implies that we can change 
the input into $T'$, and use our approximation for {\md} {\gst} on trees.

This process is very similar to the tree embedding reduction from graphs to trees,
that later uses the $O(\log^2 n)$ approximation of \cite{GKR} 
for {\sc Min-Cost} {\gst} on input trees to derive a solution. 
In both cases the reduction to trees incurs an $O(\log n)$ 
factor in the ratio, and in both cases the final ratio is $O(\log^3 n)$.

\begin{definition}
A subset $X$ of nodes in a graph $G$ is called an {\bf $\al$-balanced separator} if 
every connected component in $G \sem X$ contains at most $\al n$ nodes.
\end{definition}

It is known that any graph $G$ has $2/3$-balanced separator of size ${\tw}(G)+1$.
We will use a linear time algorithm with a slightly worse $\al$.

\begin{lemma} [\cite{R}] \label{lemma1}
There is a linear time algorithm that finds a $4/5$-balanced separator of size ${\tw}(G)+1$.
\end{lemma}

We use this to construct a tree of separators as follows,
where ${\cal L}$ will be the sets of leaves in the tree.

\medskip \medskip

\begin{algorithm}[H]
\caption{\sc{Separator-Tree}$(G=(V,E))$}
\label{alg:1}
${\cal C}\gets \{G\}$, ${\cal S}\gets \emptyset$, ${\cal L}\gets \emptyset$  \\
\While{\em there is $H \in {\cal C}$ with at least ${\tw}(G)+1$ nodes}
{
find a separator $S_H$ of $H$ of size ${\tw}(G)+1$ using Lemma \ref{lemma1} \\
add to ${\cal L}$ all connected components of size at most $k$ in $H-S_H$ \\
add to ${\cal C}$ all connected components of size at least $k+1$ in $H-S_H$ \\ 
delete $S_H$ from $H$ and add it to ${\cal S}$
}
\Return{${\cal S},{\cal L}$}
\end{algorithm}

\subsection{Connecting the separators}

For a separator $S$, let $T_S$ be the tree  of separators rooted at $S$.
We connect the separators as follows.

\begin{definition}
Consider a connected component of $H-S_H$ and its $O(1)$ size separator $S$.
Denote by  $T_S$, the tree of separators rooted by $S$.
If $S$ is a leaf (a set with less  than $k$ nodes) set $T_S=S$.
\end{definition}

Note that $T_S$ is by definition, a connected graph.
As $T_S$ is connected, there is a path in $T_S$ between any two nodes of $S$.
The leaves are also connected sets with less than $k$ nodes.

The way we connect separators is as follows.
For every $S$, choose an arbitrary node $u\in S$, and connect $u$ using $T_S$ (say, with shortest paths) 
to all the nodes of $S-u$.
Note that only nodes of $T_S$ are used. In particular 
nodes that are inside ancestors of $S$ are not used.
If we need to connect a leaf, only the edges inside $S$ are used.
A formal algorithm for connecting separators is as follows.

\medskip \medskip

\begin{algorithm}[H]
\caption{\sc{Connect-Separators}$(G=(V,E),\{S_i\})$}
\label{alg:2}
$E'\gets \emptyset$  \\
connect every leaf separator by an arbitrary spanning tree \\
for every non leaf separator $S$, choose a node $u\in S$ and 
add $k-1$ shortest paths from $u$ to $S-u$ in the graph induced by $T_S$ \\
\Return{${\cal S},{\cal L}$}
\end{algorithm}

\medskip \medskip

Note that the edges of the ancestors sets of $S$ are not used here.

\begin{lemma} \label{cons}
The above adds to the degree of every node, at most $O(k\cdot \log n)=O(\log n)$ edges, for every node.
\end{lemma}
\begin{proof}
Because of the DFS  structure of the separators 
(namely every edge that is not between parent child separator 
is a backward edge), the degree of a node can be affected 
by at most one separator $S$ per level. We add $|S|-1$ paths to connect $s$,
and the paths belong to $T_S$, the tree that $S$ roots.
Since $|S|=k=O(1)$ the number of paths per $S$ is $O(1)$.
Since a node is affected by at most one separator at each one 
of the $O(\log n)$ levels, the degree added in the procedure is $O(k\cdot \log n)=O(\log n)$.
\qed
\end{proof}

\noindent
{\bf Remark:} Let $S_R$ be the root separator. 
The fact that connecting a separator
does not require edges to ancestors of the separators, is crucial.
Indeed a level $i$ may have a very large number 
of separators. And if all of them, say, use edges of the root separator set, $S_R$, the degree in $S_R$ can not be bounded.

\medskip

We next contract each separator $S$. The groups the new node 
belongs to, is the union over all nodes in $S$ of the group that node belongs to.

\subsection{After the separators are connected}

Since the separators are connected we contract every separator into a single node, and update the groups it belong to to the union 
of the groups that each node in $S$ belongs to. We remove self loops and parallel edges.
Let $G'$ be the 
resulting new graph. 
Let $T'$ be the tree $T'$ resulting from removing 
all backward edges from $G'$.
The tree $T'$ has height $O(\log n)$.
Note that $G'$ contains all the edges of $G$, that were not contracted yet.
The contracted edges induce a low degree and are part of the solution.

We show how to modify an optimum set $OPT$ of edges,
so that the new solution will not contain backwards edges,
and the degrees increase by at most an additive factor of $O(\log n\cdot d^*)$.
Since the new solution has no backward edges 
it is a subtree of $T'$. Thus we can change the input 
to {\em $T'$ } and apply our algorithm for trees, since our modified solution is a subtree of $T'$.

The following algorithm we use information on  $OPT$, the optimum solution, that we do not know.
Since $OPT$ is not known, the proof is existential.

Denote by $r$ the root of the contracted graph.

\medskip \medskip

\begin{algorithm}[H]
\caption{\sc{Discard-Back-Edges}($T(G,r)$)} \label{alg:3}
\While{\em there is a backward edge $(r,w)$ in $OPT$}
{
Remove $(r,w)$ from $OPT$ \\
Add the path $P_{r,w}$ from $r$ to $w$ in $T'$ to $OPT$
}
\Return{\em the resulting graph}
\end{algorithm}

\medskip \medskip

\begin{algorithm}[H]
\caption{\sc{Tree-Reduce}$(OPT,u)$} \label{alg:4}
\If{\em the height of the tree is at least $2$}
{
Apply Algorithm~3 on $OPT,u$  \\
For every child $w$ of $u$ apply the algorithm recursively on $T_w$
}
\end{algorithm}

\medskip \medskip

\begin{theorem}
The above algorithm finds a solution of degree at most $O(\log^3 n)\cdot d^*$
with $d^*$ the minimum degree.
\end{theorem}
\begin{proof}
We bound the affect of the (existential) modification of our solution into a tree.
We show that the number of edges added to a node is at most $O(\log n\cdot d^*)$ with $d^*$ the minimum degree.
At every level, a node is influenced by only one separator $S$.
The size of each separator $S$ is $k=O(1)$. 
Each one of the $O(1)$ nodes inside a separator, is touched by at most $d^*\cdot O(1)=O(d^*)$ downward edges. Each downward edge, 
was replaced by a path. Thus the unique separator $S$ of a node $v$ 
in $T_S$ may add $2\cdot O(d^*)=O(d^*)$ edges to $v$.
Since the number of levels is $O(\log n)$, the change in the degree 
of every node will be an additive $O(\log n)\cdot d^*$ factor in the degree.
This proves the existence of a tree 
of maximum degree $O(\log n)\cdot d^*$. Applying our approximation 
for MDGS on trees derives an $O(\log^3 n)$ approximation ratio.
\qed
\end{proof}

\section{Proof of Theorem~\ref{t:3}}

We will assume that we know a node $r$ that belongs to some optimal solution.
We root the input tree $T$ at $r$. For $S \in \SS$ let 
$
\AA_S=\{A \subs V: r \notin A, S \subs A\}
$ 
be the family of cuts that separate the group $S$ from $r$. 
Let $\AA=\cup_{S \in \SS} \AA_S$ be the family of all cuts that separate $r$ from some group. 
The algorithm of Garg, Konjevod, and Ravi \cite{GKR} uses the following natural LP for the {\sc Min-Cost} {\gst} problem
\[\begin{array} {llllllll} 
&  \min           & c \cdot x                                        &                            &\\
&  \mbox{s.t.} & x(\de(A)) \geq 1 \hspace{0.3cm} & \forall A \in \AA &\\
&                     & x_e \geq 0                                     & \forall e  \in E    & 
\end{array}\]

A standard randomized rounding process may fail to connect even one group. 
The authors of \cite{GKR} give a special rounding method.
For $e \in E$ let $p(e)$ be the parent edge of $e$, $p^2(e)=p(p(e))$ the parent edge of $p(e)$, and so on -- 
$p^i(e)$ is the $i$th edge on the path from $e$ to the root.
Add a dummy parent edge $f$ of the root $r$ and set $x_f=1$.
The algorithm of \cite{GKR} connects a fraction of groups to the root by 
choosing every edge $e\in E$ with probability $x_e/x_{p(e)}$.
Then the probability that an edge $e$ of depth $i$ is connected to the root is 
$$
\f{x_e}{x_{p(e)}}\cdot \f{x_{p(e)}}{x_{p(p(e))}} \cdots \f{x_{p^{i-1}(e)}}{x_{p^i(e)}} \cdot \f{x_{p^i(e)}}{x_f}=\f{x_e}{x_f}=x_e \ .
$$
Thus the expected cost of the edges that are connected to $r$ is bounded by  
the value $c \cdot x$ of the LP solution. The key statement in \cite{GKR} is:

\begin{theorem} [\cite{GKR}]
The probability that a group $S$ is connected the root by the above random 
process is $\Omega(1/\log N)$.
\end{theorem}

Thus the expected number of iterations required to connect all groups to the root 
is $O(\log N\cdot \log |\SS|)$, and therefore, this is the expected approximation ratio.

For a node $v$ let $e_v$ denote the parent edge of $v$.
In our algorithm for the {\bd} {\gst} we modify the \cite{GKR} LP by adding the following degree constraints inequalities:
\begin{equation} \label{e:dc}
x(\de(v)) \leq x_{e_v}  \cdot b_v  \ \ \ \ \ \forall v \in V \ .
\end{equation}
To see that these are valid inequalities, consider the chatacteristic vector $x$ of an inclusion minimal feasible solution $T$.
If $x_{e_v}=0$ then $x(\de(v))=0$, since $v \notin T$.
If $x_{e_v}=1$ then $x(\de(v)) \leq b_v =x_{e_v}  \cdot b_v$.

We use the same rounding as \cite{GKR}.
The added degree constrains inequalities do not change the cost approximation analysis of \cite{GKR} - 
the expected number of iterations is still $O(\log N \cdot \log |\SS|)$ and so is the expected cost approximation ratio.
We will analyze the degrees approximation separately.
We use the Chernoff bound.
If $X$ is a sum of $n$ independent Bernoulli variables, with mean $\mu$, then:
\begin{eqnarray}
\label{cher1}
\Pr\left[X> (1+\delta)\mu\right]\leq {\left (\frac{e^{\delta}}{(1+\delta)^{1+\delta}}\right)}^{\mu}.
\end{eqnarray}

The degree of $v$ results by $O(\log N\cdot \log |\SS|)$ iterations.
In each round we have a Bernoulli sum, of all the children of $v$ that did not reach 
the root yet. The difficulty here is that the  random Bernoulli variables are dependent.

For simplicity of the analysis, we bound the degree by $O(\log N\cdot \log |\SS|)$ {\em independent} 
Bernoulli sums, that contains all neighbors of $v$ in every round. This random variable bounds the degree from above
as a child $u$ can contribute more than $1$ to the degree.
However our random process gives  a sum of independent Bernoulli variable 
which  makes the analysis simpler.
For a node $v$, we have a sum of $\deg(v)\cdot O(\log N\cdot \log \SS|)$
independent Bernoulli variables. The expected degree is $\tau_v=O(\log N\cdot \log |\SS|)\cdot x(\de(v))/x_{e_v}$,
and note that $x(\de(v))/x_{e_v} \leq b_v$.
We now bound the expectation of $\tau_v$ by three claims.

\begin{claim}
If $\tau_v\geq c\cdot \log n$ for some constant $c$ then 
with probability $1-1/n^2$, $\deg(v)=O(\log^2 n)\cdot  b_v$.
\end{claim}
\begin{proof}
$$
\Pr\left[\deg(v)>2\tau_v\right]\leq \left(\frac{e}{4}\right)^{c\log n} \leq \frac{1}{n^2} \ .
$$ 
The last inequality holds for large enough $c$. 
Note that this implies that with probability $1-1/n^2,$ 
$\deg(v)\leq O(\log N\cdot \log |\SS|)\cdot b_v$.
The last inequality, uses the valid inequality $x(\de(v))/x_{e_v} \leq b_v$ from the LP. 
The ratio $O(\log^2 n)$ follows.
\end{proof}

We now deal with nodes for which 
$1\leq \tau_v\leq c\cdot \log n$ for some constant $c$.

\begin{claim}
In this case, with probability at least $1-1/n^2$, $\deg(v)=O(\log^2 n)$. 
Since $b_v\geq 1$ the ratio is $O(\log^2 n)$.
\end{claim}
\begin{proof}
We know that $\tau_v\leq c\log n$.
Set $(1+\delta)=\log n$.

First we note that if we prove that 
$\Pr[\deg(v)\geq (1+\delta)\tau_v]\leq 1/n^2$, then since 
$\delta=O(\log n)$ and 
$\tau_v=O(\log n)$ we get that 
with probability $1-1/n^2$ that $\deg(v)=O(\log^2 n)$.
Since $b_v\geq 1$ this gives ratio $O(\log^2 n)$. 
We now prove the required inequality.

Since $\tau_v\geq 1$ we get from the Chernoff bound that:
$$
\Pr[\deg(v)\geq (1+\delta) \tau_v] \leq \frac{e^{\log n}}{\left (\log n\right)^{\log n }} \ .
$$
For large enough $n$ this probability is at most $1/n^2$.
\end{proof}

The last case is $\tau_v<1$.

\begin{claim}
In case $\tau_v<1$ with probability $1-1/n^2$, 
$\deg(v)=O(\log n)\cdot b_v$.
\end{claim}
\begin{proof}
We set $(1+\delta)=\log n/\tau_v$.
Note that if $\deg(v)\leq (1+\delta)\cdot \tau_v$ 
then $\deg(v)=O(\log n)$. As $b_v\geq 1$ the ratio is $O(\log n)$.
We now bound 
$$\Pr[\deg(v)> (1+\delta)\cdot \tau_v)]$$ 

Consider the term: 
$$\frac{e^{\delta}}{(1+\delta)^{1+\delta}}
	\leq \frac{e^{\log n/\tau_v}}{\left (\log n/\tau_v\right)^{\log n /\tau_v}}.$$

To get the Chernoff bound we should raise 
to above to the power $\tau_v$.
Raising this term to $\tau_v$,  
the $\tau_v$ factor cancels in both exponents. Thus:
$$\Pr\left[\deg(v) \leq (1+\delta)\tau_v\right] \leq \frac{e^{\log n}}{\left (\log n/\tau_v\right)^{\log n }}.$$
Since $\tau_v<1$ the above is bounded by 
$$
\frac{e^{\log n}}{\left (\log n\right)^{\log n }}.$$
and the above term is bounded by $1/n^2$ for large enough $n$.
\end{proof}

We got that with probability $1-1/n^2$, for a given $v$, $\deg(v)=O(\log^2 n)\cdot b_v$. 
By the union bound with probability $1-1/n$ for every $v$, $\deg(v)\leq O(\log^2 n)\cdot b_v$.


\begin{thebibliography}{10}

\bibitem{BKN}
N.~Bansal, R.~Khandekar, and V.~Nagarajan.
\newblock Additive guarantees for degree-bounded directed network design.
\newblock {\em {SIAM} J. Computing}, 39(4):1413--1431, 2009.

\bibitem{CEK}
C.~Chekuri, G.~Even, and G.~Kortsarz.
\newblock A greedy approximation algorithm for the group {Steiner} problem.
\newblock {\em Discrete Applied Mathematics}, 154(1):15--34, 2006.

\bibitem{FRT}
J.~Fakcharoenphol, S.~Rao, and K.~Talwar.
\newblock A tight bound on approximating arbitrary metrics by tree metrics.
\newblock {\em J. Computer and Syst. Sciences}, 69(3):485--497, 2004.

\bibitem{G}
N.~Garg.
\newblock Saving an epsilon: a 2-approximation for the $k$-{MST} problem in
  graphs.
\newblock In {\em STOC}, pages 396--402, 2005.

\bibitem{GKR}
N.~Garg, G.~Konjevod, and R.~Ravi.
\newblock A polylogarithmic approximation algorithm for the group {Steiner}
  tree problem.
\newblock {\em J. Algorithms}, 37(1):66--84, 2000.

\bibitem{HK}
E.~Halperin and R.~Krauthgamer.
\newblock Polylogarithmic inapproximability.
\newblock In {\em STOC}, pages 585--594, 2003.

\bibitem{J}
K.~Jain.
\newblock A factor $2$ approximation algorithm for the generalized steiner
  network problem.
\newblock {\em Combinatorica}, 21(1):39--60, 2001.

\bibitem{KKN}
R.~Khandekar, G.~Kortsarz, and Z.~Nutov.
\newblock Network-design with degree constraints.
\newblock In {\em APPROX-RANDOM}, pages 289--301, 2011.

\bibitem{BP}
B.~Laekhanukit.
\newblock Private communication.
\newblock 2019.

\bibitem{LNSM}
L.~C. Lau, J.~Naor, M.~Salavatipour, and M.~Singh.
\newblock Survivable network design with degree or order constraints.
\newblock {\em SIAM J. Computing}, 39(3):1062--1087, 2009.

\bibitem{LRS}
L.~C. Lau, R.~Ravi, and M.~Singh.
\newblock {\em Iterative Methods in Combinatorial Optimization}.
\newblock Cambridge University Press, 2011.

\bibitem{MR}
R.~Motwani and P.~Raghavan.
\newblock {\em Randomized Algorithms}.
\newblock Cambridge University Press, 1995.

\bibitem{PR}
L.~J. Poplawski and R.~Rajaraman.
\newblock Multicommodity facility location under group steiner access cost.
\newblock In {\em SODA}, pages 996--1013, 2011.

\bibitem{R}
B.~Reed.
\newblock Finding approximate separators and computing tree width quickly.
\newblock In {\em STOC}, pages 221--228, 1992.

\bibitem{SK}
N.~Sharma and M.~Kaur.
\newblock Survey of {VLSI} techniques for power optimization and estimation of
  optimization.
\newblock {\em International Journal of Emerging Technology and Advanced
  Engineering}, 4:351--355, 2014.

\bibitem{SCH}
P.~J. Slater, E.~J. Cockayne, and S.~T. Hedetniemi.
\newblock Information dissemination in trees.
\newblock {\em {SIAM} J. Comput.}, 10(4):692--701, 1981.

\bibitem{china}
Y.~Wang, X.~Hong, T.~Jing, Y.~Yang, X.~Hu, and G.~Yan.
\newblock An efficient low-degree {RMST} algorithm for {VLSI/ULSI} physical
  design.
\newblock In {\em PATMOS}, pages 442--452, 2004.

\end{thebibliography}

\end{document}